\documentclass[aps,prl,reprint,superscriptaddress,amsmath,amssymb,longbibliography]{revtex4-2}
\usepackage[T1]{fontenc}
\usepackage{lmodern}
\usepackage{mathtools}
\usepackage{amsthm}
\usepackage{graphicx}
\usepackage{microtype}
\usepackage{hyperref}
\usepackage{xcolor}

\newcommand{\D}{\mathcal D}
\newcommand{\F}{\mathcal F}
\newcommand{\Ocal}{\mathcal O}
\newcommand{\supp}{\operatorname{supp}}
\newcommand{\spn}{\operatorname{span}}
\newcommand{\ket}[1]{\lvert #1\rangle}
\newcommand{\bra}[1]{\langle #1\rvert}
\newcommand{\norm}[1]{\lVert #1\rVert}
\newcommand{\Tr}{\operatorname{Tr}}

\newcommand{\xto}{\longrightarrow}
\newcommand{\notto}{\nrightarrow}

\newtheorem{theorem}{Theorem}
\newtheorem{lemma}{Lemma}

\hypersetup{colorlinks=true,
 linkcolor=magenta,
 citecolor=magenta,
 urlcolor=magenta,
 filecolor=magenta}

\begin{document}
\title{Total Ordering of Resources: A Quantum--Classical Separation}
\author{Jingsong Ao}
\affiliation{Institute of Fundamental Technological Research, Polish Academy of Sciences, Pawi\'nskiego 5B, 02-106 Warsaw, Poland}
\author{Julio I. de Vicente}
\affiliation{Departamento de Matem\'aticas, Universidad Carlos III de Madrid, E-28911, Legan\'es (Madrid), Spain}
\affiliation{Instituto de Ciencias Matem\'aticas (ICMAT), E-28049 Madrid, Spain}
\author{Alexander Streltsov}
\affiliation{Institute of Fundamental Technological Research, Polish Academy of Sciences, Pawi\'nskiego 5B, 02-106 Warsaw, Poland}

\begin{abstract}

A central goal in quantum information science is understanding which physical states can be converted into others. The ideal scenario would be a totally ordered system, where every state can be ranked by a single value, and a higher-ranked state can always be converted into a lower-ranked one. Such an ordering would provide a complete description of state convertibility, making the resource structure maximally simple to characterize and compare. We prove that such a simple ordering is fundamentally impossible for any quantum system with dimension larger than~2. The underlying properties of quantum mechanics guarantee that some states will always be incomparable. Furthermore, we show this is a uniquely quantum phenomenon by constructing classical systems of any size that can be perfectly ranked. This establishes a rigorous and striking divide between quantum and classical state transformations.
\end{abstract}
\maketitle

\emph{Introduction.---}Entanglement is one of the defining resources of quantum information: shared entangled states enable information-processing tasks and advantages that cannot be achieved without entanglement under the same operational restrictions~\cite{HorodeckiReview,BennettTeleportation}. This perspective has developed into the broader framework of quantum resource theories, in which one specifies states that are available for free and transformations that cannot generate the resource~\cite{ChitambarGour,CoeckeFritzSpekkens}. The same language encompasses resources as diverse as entanglement~\cite{HorodeckiReview}, coherence~\cite{Baumgratz2014,CoherenceReview}, purity~\cite{Horodecki2003,GourNonequilibrium}, and quantum thermodynamics~\cite{Brandao2015}, allowing structural questions about quantum advantage to be studied independently of a particular physical implementation.

A central question in any quantum resource theory is state convertibility: given two states, can one be transformed into the other using only free operations? Complete conversion criteria have been formulated using families of entropic monotones~\cite{Gour2017,Gour2018} and operational quantities based on discrimination tasks~\cite{TakagiRegula2019} or work extraction~\cite{HsiehGessner2026}. In many resource theories some pairs of states are incomparable, meaning that neither state can be converted into the other by free operations. For bipartite pure-state entanglement, for example, deterministic transformations under local operations and classical communication are governed by majorization, which already in local dimension $d\ge3$ gives rise to incomparable pairs~\cite{Nielsen,Vidal}. 

The opposite extreme is a totally ordered resource theory, in which every pair of states is freely convertible in at least one direction. Such a theory admits a particularly simple description of resource conversion, because a single quantity can capture the entire conversion structure~\cite{Datta} therefore defining resource value univocally. On the other hand, resource theories with incomparable states give rise to a complicated zoo of different monotones, which reflects the fact that resource value can become task-dependent according to the different operational interpretations different monotones might have.

Incomparability also occurs for pure-state coherence under incoherent operations~\cite{Zhu2017} and for purity~\cite{Horodecki2003,GourNonequilibrium}, whose conversion criteria likewise involve majorization. At the same time, nontrivial totally ordered quantum resource theories are known to exist for qubits, while whether they can exist in dimensions $d\ge3$ was left open in~\cite{Datta}. The presence of incomparable states in commonly used quantum resource theories does not by itself show that incomparability is an intrinsic feature of quantum state space: it could instead arise from a particular choice of free states or free operations. 

In this letter we prove that for every nontrivial quantum resource theory total ordering is impossible in dimension $d\ge3$, even when one allows the maximal set of free-state-preserving quantum channels. In contrast, we construct nontrivial classical resource theories that are totally ordered in every finite dimension $d\ge2$. This means that the nonexistence of total order is a purely quantum phenomenon.

\medskip

\emph{Totally ordered quantum resource theories.---}Let $\D_d$ be the density operators on $\mathbb C^d$. Fix a nonempty compact convex set $\F\subsetneq\D_d$ and define
\begin{equation}
 \Ocal_{\F}:=\{\Lambda\text{ CPTP}:\Lambda(\F)\subseteq\F\}.
 \label{eq:maximal}
\end{equation}
The elements of $\F$ are the \emph{free states}, available without consuming the resource; states outside $\F$ are resource states. We take $\Ocal_{\F}$ to be the \emph{free operations}, namely all channels that cannot create a resource state from a free input~\cite{ChitambarGour}. For bipartite entanglement, $\F$ consists of separable states, i.e., mixtures of product states, and $\Ocal_{\F}$ is the set of separability-preserving channels, which contains local operations and classical communication~\cite{HorodeckiReview,ChitambarGour}. For coherence in a fixed reference basis, $\F$ consists of diagonal density operators and $\Ocal_{\F}$ is the set of maximally incoherent operations~\cite{CoherenceReview,Aberg2006,ChitambarGourCoherence}. For purity, $\F=\{I_d/d\}$ and $\Ocal_{\F}$ consists of unital channels~\cite{GourNonequilibrium}. These maximal sets need not coincide with the more restrictive operations motivated by a particular experimental setting. We write $\rho\xto\sigma$ if $\Lambda(\rho)=\sigma$ for some $\Lambda\in\Ocal_{\F}$. 

Nonemptiness and properness ensure that both free and resource states exist. Convexity assumes that probabilistic mixtures of free preparations remain free. Closedness includes limits of free preparations, so a state approximable arbitrarily well by free states is itself free. Since $\D_d$ is compact, closedness of $\F$ is equivalent to compactness. These assumptions are standard in the framework of quantum resource theories~\cite{ChitambarGour,Datta}.

We call a resource theory \emph{totally ordered} if, for every pair $\rho,\sigma\in\D_d$, at least one of the conversions $\rho\xto\sigma$ or $\sigma\xto\rho$ is possible~\cite{Datta}. Equivalently, the conversion preorder has no incomparable pairs. Totality is particularly restrictive because it means that the complete convertibility structure can be represented by a single monotone: states can be placed on one resource scale rather than requiring several independent monotones. A resource monotone is a real-valued function $M$ that cannot increase under free operations, i.e., $\rho\xto\sigma$ implies $M(\rho)\ge M(\sigma)$. Nontrivial totally ordered quantum resource theories do exist for $d=2$; for example, qubit purity with free state $I_2/2$ and unital channels is totally ordered by the Bloch-vector length~\cite{Datta}.

\medskip

\emph{No total order beyond qubits.---}Our aim is to show that, for every admissible set of free states $\F$ in dimension $d\ge3$, some pair of states admits neither conversion even under the maximal choice of free operations in Eq.~\eqref{eq:maximal}. Enlarging a set of free operations can only add conversions. An incomparable pair under $\Ocal_{\F}$ therefore remains incomparable under every smaller set preserving the same $\F$. The obstruction is consequently independent of any additional operational restriction; it is not merely the failure of totality in a particular example.

\begin{theorem}\label{thm:quantum}
For every $d\ge3$ and every nonempty compact convex $\F\subsetneq\D_d$, there exist $\rho,\sigma\in\D_d$ such that neither $\rho\xto\sigma$ nor $\sigma\xto\rho$ is possible under $\Ocal_{\F}$. Consequently, no nontrivial quantum resource theory is totally ordered for $d\ge3$.
\end{theorem}
\noindent The situation is illustrated in Fig.~\ref{fig:incomparability}.

\begin{figure}[t]
\centering
\includegraphics[width=\columnwidth]{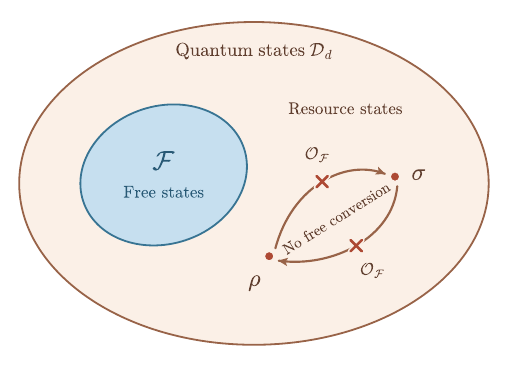}
\caption{Schematic representation of quantum incomparability. The free states form a proper subset $\F\subsetneq\D_d$. For $d\ge3$, Theorem~\ref{thm:quantum} guarantees resource states $\rho$ and $\sigma$ for which neither conversion direction is possible under the maximal free-operation set $\Ocal_{\F}$.}
\label{fig:incomparability}
\end{figure}

In particular, there is no resource monotone $M:\D_d\to\mathbb R$ satisfying
\begin{equation}
 \rho\xto\sigma\quad\Longleftrightarrow\quad M(\rho)\ge M(\sigma)
 \label{eq:complete}
\end{equation}
for all states. This conclusion does not require $M$ to be continuous, convex, or to satisfy any other regularity condition.

Suppose, towards a contradiction, that $\Ocal_{\F}$ totally orders $\D_d$, with $d\ge3$. We use two structural facts established in Ref.~\cite{Datta}: all pure states are freely interconvertible, and a nontrivial totally ordered qubit theory has free-state set
\begin{equation}
 \F_t=\{\rho\in\D_2:\norm{\rho-I_2/2}_1\le t\},\qquad 0\le t<1.
 \label{eq:qubit}
\end{equation}
Here $\norm{A}_1:=\Tr\sqrt{A^\dagger A}$ denotes the trace norm.
No pure state can therefore be free, since convexity would then make every state free. The continuous convex monotone
\begin{equation}
 R(\rho):=\min_{\mu\in\F}\norm{\rho-\mu}_1
 \label{eq:distance}
\end{equation}
has a common value $R(\psi)=r>0$ on all pure states. Here and in the following, given a unit-norm vector $|\psi\rangle$, we write $\psi=\ket\psi\!\bra\psi$ for the corresponding pure-state density matrix. Similarly, monotonicity of $\max_{\mu\in\F}F(\rho,\mu)$, with squared fidelity $F(\rho,\mu)=\norm{\sqrt\rho\sqrt\mu}_1^2$, implies that a constant $0<x<1$ satisfies
\begin{equation}
 \max_{\mu\in\F}\bra\psi\mu\ket\psi=x
 \qquad\text{for every pure }\psi.
 \label{eq:saturation}
\end{equation}
This implies $\bra\psi\mu\ket\psi\le x$ for every pure $\psi$ and every free state $\mu\in\F$, and hence
\begin{equation}
 \lambda_{\max}(\mu)\le x\qquad \forall \mu\in\F.
\end{equation}
Here $\lambda_{\max}(\mu)$ denotes the largest eigenvalue of $\mu$. In what follows we will also use that $\Ocal_{\F}$ is compact, which follows from the compactness of $\F$.

For any subspace $X\subseteq\mathbb C^d$, let $P_X$ be the projector onto $X$ and define $\D_X:=\{\rho\in\D_d:\supp\rho\subseteq X\}$ and $\F_X:=\F\cap\D_X$. Two consequences of totality drive the proof, which are presented in Lemmas~\ref{lem:restriction} and \ref{lem:ranktwo}.

\begin{lemma}\label{lem:restriction}
Assume the nontrivial quantum resource theory is totally ordered and let $W$ be a subspace. If $\F_W\ne\varnothing$, the restrictions of global free channels mapping $\D_W$ into itself define a nontrivial totally ordered theory on $\D_W$. Moreover, if $\dim W=2$, then
\begin{equation}
 R(P_W/2)=\begin{cases}
 0,&\text{if }\F_W\ne\varnothing,\\
 r>0,&\text{if }\F_W=\varnothing.
 \end{cases}
 \label{eq:dichotomy}
\end{equation}
\end{lemma}

\begin{proof}
Let $\rho\in\D_W$ be full rank on $W$, let $\Lambda$ be positive and trace preserving, let $V\subseteq\mathbb C^d$ be a subspace, and suppose $\supp\Lambda(\rho)\subseteq V$. Writing $m=\lambda_{\min}(\rho|_W)>0$, where $\rho|_W$ denotes the restriction of $\rho$ to $W$ and $\lambda_{\min}$ its smallest eigenvalue, every $\omega\in\D_W$ satisfies
\begin{equation}
 0\le\omega\le P_W\le m^{-1}\rho,
 \qquad 0\le\Lambda(\omega)\le m^{-1}\Lambda(\rho).
 \label{eq:supportlemma}
\end{equation}
Hence $\Lambda(\omega)$ is supported on $V$, and trace preservation gives $\Lambda(\D_W)\subseteq\D_V$.

Now suppose $\F_W\ne\varnothing$. Let $\Ocal_W$ consist of the restricted channels $\Lambda|_W$ obtained from all $\Lambda\in\Ocal_{\F}$ satisfying $\Lambda(\D_W)\subseteq\D_W$. Here $\Lambda|_W$ means the same linear map $\Lambda$ with its input restricted to operators supported on $W$; the condition $\Lambda(\D_W)\subseteq\D_W$ ensures that it is a well-defined CPTP map on the Hilbert space $W$. These restricted channels preserve $\F_W$. They contain the identity and the replacement channels with output in $\F_W$, and are closed under composition. To establish totality, take arbitrary $\rho,\sigma\in\D_W$, put $k=\dim W$, and construct the following states:
\begin{equation}
 \rho_n=(1-\varepsilon_n)\rho+\varepsilon_n P_W/k,\quad
 \sigma_n=(1-\varepsilon_n)\sigma+\varepsilon_n P_W/k,
 \label{eq:regularization}
\end{equation}
with $\varepsilon_n\downarrow0$. Totality supplies a conversion in at least one direction for each $n$. One direction occurs infinitely often; along this subsequence suppose $\Lambda_n(\rho_n)=\sigma_n$. Since $\rho_n$ is full rank on $W$, Eq.~\eqref{eq:supportlemma} implies $\Lambda_n(\D_W)\subseteq\D_W$.

Compactness of $\Ocal_{\F}$ gives a further subsequence with $\Lambda_n\to\Lambda_0\in\Ocal_{\F}$. Closedness of $\D_W$ and continuity yield $\Lambda_0(\D_W)\subseteq\D_W$ and $\Lambda_0(\rho)=\sigma$. Thus the restricted channel $\Lambda_{0,W}:=\Lambda_0|_W$ belongs to $\Ocal_W$. If the other conversion direction occurs infinitely often, the same reasoning gives $\sigma\xto\rho$ within $W$. The restricted theory is nontrivial because the original theory has no free pure states.

The preceding argument therefore allows us to invoke the qubit classification whenever $\dim W=2$ and $\F_W\ne\varnothing$: in that case the global theory induces a nontrivial totally ordered qubit theory on $\D_W$. If instead $\F_W=\varnothing$, there is no free state supported on $W$, so this restricted qubit theory is not available and that case must be treated separately.

For $\dim W=2$ and $\F_W\ne\varnothing$, the qubit characterization~\eqref{eq:qubit} gives $P_W/2\in\F_W$, hence $R(P_W/2)=0$. If $\F_W=\varnothing$, choose any pure $\psi$ and define the full-rank state
\begin{equation}
 \gamma_\eta=(1-\eta)\psi+\eta I_d/d,\qquad 0<\eta<1.
 \label{eq:gamma}
\end{equation}
A free conversion $\gamma_\eta\xto P_W/2$ would, by~\eqref{eq:supportlemma}, send the entire state space into $\D_W$. In particular it would send a free state into the empty set $\F\cap\D_W$, a contradiction. Totality therefore implies $P_W/2\xto\gamma_\eta$. Monotonicity and continuity give
\begin{equation}
 R(P_W/2)\ge\lim_{\eta\downarrow0}R(\gamma_\eta)=r.
 \label{eq:dichotomylower}
\end{equation}
If $w_1,w_2$ are orthogonal pure states spanning $W$, convexity gives the reverse bound:
\begin{equation}
 R(P_W/2)=R\bigl((w_1+w_2)/2\bigr)
 \le\bigl(R(w_1)+R(w_2)\bigr)/2=r.
 \label{eq:dichotomyupper}
\end{equation}
This proves Lemma~\ref{lem:restriction}.
\end{proof}

Lemma~\ref{lem:restriction} shows that totality imposes a sharp dichotomy on every two-dimensional subspace, depending on whether it contains a free state. To exploit this dichotomy in the qutrit argument, we also need a structural consequence of totality that guarantees the existence of suitable low-rank free states. The next lemma provides precisely this ingredient by producing a rank-two free state through any prescribed pure state.

\begin{lemma}\label{lem:ranktwo}
Assume $d\ge2$, $\F\subsetneq\D_d$ is nonempty, compact, and convex, and $\Ocal_{\F}$ totally orders $\D_d$. Let $x\in(0,1)$ be the constant defined by Eq.~\eqref{eq:saturation}. Then, for every pure $\psi_1$, there is a pure $\psi_2$ orthogonal to $\psi_1$ such that
\begin{equation}
 x\psi_1+(1-x)\psi_2\in\F.
 \label{eq:ranktwo}
\end{equation}
\end{lemma}
\noindent We refer to the End Matter for the proof of the lemma.

We now prove Theorem~\ref{thm:quantum} for $d=3$ assuming for a contradiction that $\Ocal_{\F}$ totally orders $\D_3$. By Lemma~\ref{lem:ranktwo}, choose orthogonal pure states $a,b$ with $xa+(1-x)b\in\F$, and set $W_{12}=\spn\{\ket a,\ket b\}$. Choose $\ket{e_3}\perp W_{12}$. By Lemma~\ref{lem:ranktwo}, some $\ket{e_1}\perp \ket{e_3}$ satisfies $xe_3+(1-x)e_1\in\F$. Since $d=3$, $\ket{e_1}\in W_{12}$; complete it to an orthonormal basis $\{\ket{e_1},\ket{e_2}\}$ of $W_{12}$. Both $W_{12}$ and $W_{13}=\spn\{\ket{e_1},\ket{e_3}\}$ contain free states. By Lemma~\ref{lem:restriction} and~\eqref{eq:qubit}, their restricted free sets are centered balls and hence invariant under all unitaries within the respective two-dimensional subspaces. More explicitly, Lemma~\ref{lem:restriction} shows that the resource theory restricted to each of these subspaces is a nontrivial totally ordered qubit resource theory, so the qubit characterization~\eqref{eq:qubit} applies to each restriction. In particular,
\begin{equation}
 \sigma_{1j}:=xe_1+(1-x)e_j\in\F,\qquad j=2,3.
 \label{eq:twofree}
\end{equation}

For small $\varepsilon>0$, compare
\begin{align}
 \rho_\varepsilon&=(1-\varepsilon)e_1+\frac{\varepsilon}{2}(e_2+e_3),\nonumber\\
 \tau_\varepsilon&=(1-2\varepsilon)e_1+2\varepsilon e_2.
 \label{eq:qutritpair}
\end{align}
Suppose a free channel $\Gamma$ maps $\tau_\varepsilon$ to $\rho_\varepsilon$, and set $\alpha=\bra{e_1}\Gamma(e_1)\ket{e_1}$, $\beta=\bra{e_1}\Gamma(e_2)\ket{e_1}$. Since $\alpha,\beta\in[0,1]$, $\bra{e_1}\Gamma(\tau_\varepsilon)\ket{e_1}=\bra{e_1}\rho_\varepsilon\ket{e_1}$ gives 
\begin{equation}
    (1-2\varepsilon)\alpha+2\varepsilon\beta=1-\varepsilon.
\end{equation}
Using $\alpha,\beta\le1$ we obtain
\begin{align}
 \alpha\ge\frac{1-3\varepsilon}{1-2\varepsilon},\qquad &\beta\ge\frac12.
 \label{eq:bounds}
\end{align}
The free state $\sigma_{12}$ would consequently satisfy
\begin{equation}
 \bra{e_1}\Gamma(\sigma_{12})\ket{e_1}
 \ge x\frac{1-3\varepsilon}{1-2\varepsilon}+\frac{1-x}{2}>x
 \label{eq:forbidden}
\end{equation}
for small enough $\varepsilon$, contradicting~\eqref{eq:saturation}. Thus $\tau_\varepsilon\notto\rho_\varepsilon$, and totality supplies a free $\Lambda_\varepsilon$ with $\Lambda_\varepsilon(\rho_\varepsilon)=\tau_\varepsilon$.

Take a convergent subsequence as $\varepsilon\downarrow0$, with limit $\Lambda_0\in\Ocal_{\F}$. Since $\rho_\varepsilon$ is full rank and $\tau_\varepsilon$ is supported on $W_{12}$, the support constraint established after Eq.~\eqref{eq:supportlemma} gives
\begin{equation}
 \Lambda_0(\D_3)\subseteq\D_{W_{12}},\qquad \Lambda_0(e_1)=e_1.
 \label{eq:compression}
\end{equation}
For $j=2,3$, preservation of the free states~\eqref{eq:twofree} yields $\Lambda_0(\sigma_{1j})\in\F$, so Eq.~\eqref{eq:saturation} gives $\bra{e_1}\Lambda_0(\sigma_{1j})\ket{e_1}\le x$. Using $\sigma_{1j}=xe_1+(1-x)e_j$, linearity of $\Lambda_0$, and $\Lambda_0(e_1)=e_1$ from~\eqref{eq:compression}, this becomes
\begin{equation}
 x\ge x+(1-x)\bra{e_1}\Lambda_0(e_j)\ket{e_1}.
 \label{eq:zerofidelity}
\end{equation}
The overlaps $\bra{e_1}\Lambda_0(e_j)\ket{e_1}$ for $j=2,3$ therefore vanish. Positivity and~\eqref{eq:compression} force
\begin{equation}
 \Lambda_0(e_2)=\Lambda_0(e_3)=e_2.
 \label{eq:collapse}
\end{equation}
If $W_{23}=\spn\{\ket{e_2},\ket{e_3}\}$ contained a free state, Lemma~\ref{lem:restriction} would make $(e_2+e_3)/2$ free. Equation~\eqref{eq:collapse} would map it to the pure resource state $e_2$. Hence $\F_{W_{23}}=\varnothing$, whereas $\F_{W_{12}}\ne\varnothing$.

These two possibilities cannot coexist. Define
\begin{align}
 \ket{u_t}&=\cos\frac{\pi t}{2}\ket{e_3}+\sin\frac{\pi t}{2}\ket{e_1},\nonumber\\
 W_t&=\spn\{\ket{e_2},\ket{u_t}\},\qquad
 \omega_t=\frac{e_2+u_t}{2},\quad 0\le t\le1.
 \label{eq:path}
\end{align}
Then $W_0=W_{23}$ and $W_1=W_{12}$. Equation~\eqref{eq:dichotomy} implies
\begin{equation}
 R(\omega_t)\in\{0,r\},\qquad R(\omega_0)=r,\quad R(\omega_1)=0.
 \label{eq:contradiction}
\end{equation}
But $t\mapsto R(\omega_t)$ is continuous, and a continuous map from $[0,1]$ into $\{0,r\}$ must be constant. This contradiction proves Theorem~\ref{thm:quantum} for the qutrit case.

To prove the theorem for $d>3$, suppose for a contradiction that a nontrivial totally ordered resource theory exists for some $d>3$. We can use Lemmas~\ref{lem:restriction} and~\ref{lem:ranktwo} to induce a nontrivial totally ordered resource theory on a three-dimensional subspace. In particular, Lemma~\ref{lem:ranktwo} provides a rank-two free state. Embed its support in a three-dimensional subspace $V$. Lemma~\ref{lem:restriction} then produces a nontrivial totally ordered qutrit theory on $\D_V$, contradicting the result just proved. This completes the proof of Theorem~\ref{thm:quantum} for all $d \geq 3$.

\medskip

\emph{Total order for a classical resource theory.---}The above results show that total ordering is impossible for every nontrivial quantum resource theory in dimension $d\ge3$. We now show that this effect is genuinely quantum: on a classical $d$-level system, nontrivial totally ordered resource theories exist for every $d\ge2$. We give an explicit family for which a single scalar monotone completely characterizes all state conversions.

Let $\Delta_d=\{p\in\mathbb R^d:p_i\ge0,\ \sum_i p_i=1\}$, with column-stochastic matrices as classical channels. By a classical resource theory on a $d$-level system we mean a choice of free probability distributions in $\Delta_d$ together with stochastic maps that preserve the free set; distributions outside the free set are resource states.

For $d\ge2$ and $0<a\le1$, define the free-state set and the free-operation set, respectively, by
\begin{align}
 \F_a^{\mathrm{cl}}&:=\{p\in\Delta_d:p_1\ge a\},\nonumber\\
 \Ocal_a^{\mathrm{cl}}&:=\{T\text{ stochastic}:T(\F_a^{\mathrm{cl}})\subseteq\F_a^{\mathrm{cl}}\}.
 \label{eq:classicalfree}
\end{align}
For every $0<a\le1$ the free set is nonempty, compact, convex, and proper.

As we will now see, the resource theory defined by $\F_a^{\mathrm{cl}}$ and $\Ocal_a^{\mathrm{cl}}$ is totally ordered. A complete monotone is
\begin{equation}
 \begin{aligned}
 R_a(p)&:=\max\{a-p_1,0\},\\
 p\xto q&\ \Longleftrightarrow\ R_a(p)\ge R_a(q).
 \end{aligned}
 \label{eq:classicalorder}
\end{equation}
To prove this, first note that the monotone  $R_a$ can be written as the $\ell_1$-distance to the free set
\begin{equation}
 R_a(p)=\frac12\min_{f\in\F_a^{\mathrm{cl}}}\norm{p-f}_{\ell_1}
 \label{eq:tv}
\end{equation}
with $\ell_1$-norm $\norm{v}_{\ell_1}:=\sum_i|v_i|$.
For $p_1<a$ and $f\in\F_a^{\mathrm{cl}}$, we have $f_1\ge a>p_1$ and, by normalization, $\sum_{j=2}^d(p_j-f_j)=f_1-p_1$. Hence
\[
 \norm{p-f}_{\ell_1}
 =(f_1-p_1)+\sum_{j=2}^d|p_j-f_j|
 \ge 2(f_1-p_1)\ge2(a-p_1).
\]
The bound is attained by $f_1=a$ and $f_j=(1-a)p_j/(1-p_1)$ for $j\ge2$; for free $p$ the distance is zero. The $\ell_1$ norm is contractive under every stochastic matrix $T$, since for any vector $v$,
\[
\norm{Tv}_{\ell_{1}}=\sum_{i}\left|\sum_{j}T_{ij}v_{j}\right|\le\sum_{i,j}T_{ij}|v_{j}|=\sum_{j}|v_{j}|=\norm{v}_{\ell_{1}}.
\]
Together with $T(\F_a^{\mathrm{cl}})\subseteq\F_a^{\mathrm{cl}}$ for every free operation $T\in\Ocal_a^{\mathrm{cl}}$, this proves necessity in~\eqref{eq:classicalorder}.

For sufficiency, if $q\in\F_a^{\mathrm{cl}}$, the replacement map with output $q$ is free and, trivially, $p\xto q$ holds. If $q\notin\F_a^{\mathrm{cl}}$, then $q_1<a$, and the assumed inequality $R_a(p)\ge R_a(q)$ in~\eqref{eq:classicalorder} gives $p_1\le q_1<a$, so $p_1<1$. Set
\begin{equation}
 r_{p,q}:=\frac{q-p_1e_1}{1-p_1},\qquad
 Te_1=e_1,\quad Te_j=r_{p,q}\quad(j\ge2),
 \label{eq:classicalmap}
\end{equation}
where $e_j$ denotes the $j$th standard basis vector and the above inequalities guarantee that $r_{p,q}\in\Delta_d$. Since every column of $T$ is either $e_1$ or $r_{p,q}$, all entries of $T$ are nonnegative and every column sums to one; hence $T$ is column-stochastic, and
\begin{equation}
 Tp=p_1e_1+(1-p_1)r_{p,q}=q.
 \label{eq:attainment}
\end{equation}
For every $z\in\Delta_d$, the definition of the columns of $T$ gives
\[
 Tz=z_1e_1+\sum_{j=2}^d z_j r_{p,q}=z_1e_1+(1-z_1)r_{p,q}.
\]
Therefore $(Tz)_1=z_1+(1-z_1)(r_{p,q})_1\ge z_1$. Thus $T$ preserves $\F_a^{\mathrm{cl}}$ and is free. This proves the converse implication in~\eqref{eq:classicalorder}, and hence establishes that the classical resource theory is totally ordered.

\medskip

\emph{Conclusions.---}We have shown that total resource ordering exhibits a sharp quantum--classical separation. For every nontrivial quantum resource theory with a nonempty compact convex free-state set, total ordering is impossible in dimensions $d\ge3$: incomparable states are unavoidable, and no single monotone can completely characterize all deterministic state conversions. In contrast, classical systems admit nontrivial totally ordered resource theories in every dimension $d\ge2$. Thus, within the stated assumptions, total order is possible classically and for qubits, but impossible for higher-dimensional quantum systems.

The result concerns the full quantum state space and deterministic, unassisted transformations. It does not exclude total order on restricted state families or in asymptotic, catalytic, or postselected settings. Several directions remain open. It would be interesting to determine whether analogous no-go results hold in more general probabilistic theories~\cite{Barrett2007,TakagiRegula2019}, and under which conditions total ordering can re-emerge when the operational setting is enlarged. Such examples are already known: bipartite pure-state entanglement becomes totally ordered in the asymptotic regime by the entropy of entanglement~\cite{Bennett1996,HorodeckiReview}. Approximate catalytic coherence under dephasing-covariant operations is characterized by the relative entropy of coherence~\cite{Datta}, while correlated-catalytic thermodynamics under Gibbs-preserving operations is characterized by the nonequilibrium free energy when arbitrarily small output error and correlations are allowed~\cite{Shiraishi2021}. Another natural question is how many scalar monotones are minimally required to characterize convertibility once a single complete monotone is ruled out~\cite{Gour2005,Datta}.

\textit{Acknowledgments.}---We thank Chandan Datta, Ray Ganardi, Tulja Varun Kondra, and Aby Philip for discussions. This work was supported by the National Science Centre Poland (Grant No.\ 2022/46/E/ST2/00115 and 2024/55/B/ST2/01590), the Spanish Ministerio de Ciencia, Innovaci\'on y Universidades (grant No.\ PID2023-146758NB-I00, grant No.\ PID2024-160539NB-I00, and ``Severo Ochoa Programme for Centres of Excellence'' grant No.\ CEX2023-001347-S all funded by MICIU/AEI/10.13039/501100011033) and from Comunidad de Madrid (grant QUITEMAD-CM TEC-2024/COM-84).

\emph{AI~Use~Disclosure.} The derivation of the proof of Theorem \ref{thm:quantum} has greatly benefited from interactions with OpenAI's ChatGPT 5.6 Sol that were substantially revised by the authors in combination with their own independent ideas. ChatGPT 5.6 Sol and ChatGPT 6 Astra has also been used as an assistive tool for language editing, manuscript organization, and bibliographic checking. All AI-assisted suggestions were critically assessed, and the final proofs, references, statements, and conclusions were independently verified by the authors, who take full responsibility for the content of the manuscript.

\bibliography{literature}

\section*{End Matter}
\appendix
\setcounter{equation}{0}
\renewcommand{\theequation}{A\arabic{equation}}

In the following we will present the proof of Lemma~\ref{lem:ranktwo}. Fix a pure $\psi_1$. Compactness supplies a free state $\mu$ attaining the maximum in~\eqref{eq:saturation}. Since $\bra{\psi_1}\mu\ket{\psi_1}=x$ and $\lambda_{\max}(\mu)\le x$, the vector $\ket{\psi_1}$ is an eigenvector of $\mu$ with eigenvalue $x$. Choose an orthonormal eigenbasis such that
\begin{equation}
 \mu=x\psi_1+\sum_{j=2}^{d}\lambda_j\psi_j,\qquad \lambda_2>0.
 \label{eq:muspectral}
\end{equation}
Such a positive $\lambda_2$ exists because $x<1$. For small $\varepsilon>0$, let
\begin{align}
 \rho_\varepsilon&=(1-\varepsilon)\psi_1+
 \frac{\varepsilon}{d-1}\sum_{j=2}^{d}\psi_j,\nonumber\\
 \tau_\varepsilon&=(1-2\varepsilon)\psi_1+2\varepsilon\psi_2.
 \label{eq:ranktwopair}
\end{align}
Suppose a free operation $\Gamma\in\Ocal_{\F}$ satisfies $\Gamma(\tau_\varepsilon)=\rho_\varepsilon$. Define $a_1=\bra{\psi_1}\Gamma(\psi_1)\ket{\psi_1}$ and $a_2=\bra{\psi_1}\Gamma(\psi_2)\ket{\psi_1}$. Taking the $\psi_1$ overlap of $\Gamma(\tau_\varepsilon)=\rho_\varepsilon$ gives 
\begin{equation}
    (1-2\varepsilon)a_1+2\varepsilon a_2=1-\varepsilon.
\end{equation}
Since $0\le a_1,a_2\le1$, we obtain
\begin{equation}
 a_1\ge\frac{1-3\varepsilon}{1-2\varepsilon},\qquad a_2\ge\frac12.
 \label{eq:ranktwobounds}
\end{equation}
By linearity and the spectral decomposition~\eqref{eq:muspectral},
\[
 \bra{\psi_1}\Gamma(\mu)\ket{\psi_1}
 =x a_1+\sum_{j=2}^{d}\lambda_j
 \bra{\psi_1}\Gamma(\psi_j)\ket{\psi_1}.
\]
Every term in the sum is nonnegative because $\Gamma(\psi_j)$ is positive. Keeping only the $j=2$ term, for which the overlap is $a_2$, and using~\eqref{eq:ranktwobounds} gives
\begin{equation}
 \bra{\psi_1}\Gamma(\mu)\ket{\psi_1}
 \ge x\frac{1-3\varepsilon}{1-2\varepsilon}+\frac{\lambda_2}{2}>x
 \label{eq:ranktwoexclude}
\end{equation}
for sufficiently small $\varepsilon$. This contradicts~\eqref{eq:saturation}. Hence no free operation can map $\tau_\varepsilon$ to $\rho_\varepsilon$.

Totality consequently gives $\Lambda_\varepsilon\in\Ocal_{\F}$ with $\Lambda_\varepsilon(\rho_\varepsilon)=\tau_\varepsilon$. Take a convergent subsequence $\Lambda_{\varepsilon_n}\to\Lambda_0\in\Ocal_{\F}$. For every $n$, the input $\rho_{\varepsilon_n}$ is full rank on $\mathbb C^d$, while the output $\tau_{\varepsilon_n}$ is supported on $\spn\{\ket{\psi_1},\ket{\psi_2}\}$. Hence~\eqref{eq:supportlemma} gives $\Lambda_{\varepsilon_n}(\D_d)\subseteq\D_{\spn\{\ket{\psi_1},\ket{\psi_2}\}}$ for every $n$, and closedness of this state space preserves the inclusion in the limit. Moreover, $\rho_{\varepsilon_n}\to\psi_1$ and $\tau_{\varepsilon_n}\to\psi_1$; taking the limit in $\Lambda_{\varepsilon_n}(\rho_{\varepsilon_n})=\tau_{\varepsilon_n}$ therefore gives $\Lambda_0(\psi_1)=\psi_1$. Thus
\begin{equation}
 \Lambda_0(\D_d)\subseteq\D_{\spn\{\ket{\psi_1},\ket{\psi_2}\}},
 \qquad \Lambda_0(\psi_1)=\psi_1.
 \label{eq:ranktwolimit}
\end{equation}
Because $\Lambda_0(\mu)$ is free, Eq.~\eqref{eq:saturation} gives
\[
 \bra{\psi_1}\Lambda_0(\mu)\ket{\psi_1}\le x.
\]
On the other hand, linearity, Eq.~\eqref{eq:muspectral}, and $\Lambda_0(\psi_1)=\psi_1$ from~\eqref{eq:ranktwolimit} give
\[
 \bra{\psi_1}\Lambda_0(\mu)\ket{\psi_1}
 =x+\sum_{j=2}^{d}\lambda_j
 \bra{\psi_1}\Lambda_0(\psi_j)\ket{\psi_1}.
\]
Combining the two relations yields
\begin{equation}
 x\ge x+\sum_{j=2}^{d}\lambda_j
 \bra{\psi_1}\Lambda_0(\psi_j)\ket{\psi_1}.
 \label{eq:ranktwozero}
\end{equation}
Since every term in the sum in~\eqref{eq:ranktwozero} is nonnegative, the inequality can hold only if
\[
 \bra{\psi_1}\Lambda_0(\psi_j)\ket{\psi_1}=0
\]
for every $j$ with $\lambda_j>0$. For such a $j$, Eq.~\eqref{eq:ranktwolimit} shows that the state $\Lambda_0(\psi_j)$ is supported on $\spn\{\ket{\psi_1},\ket{\psi_2}\}$. Thus $\Lambda_0(\psi_j)=\psi_2$ for every $j$ with $\lambda_j>0$. Therefore
\begin{equation}
 \Lambda_0(\mu)=x\psi_1+(1-x)\psi_2\in\F,
 \label{eq:ranktwofinal}
\end{equation}
which proves Lemma~\ref{lem:ranktwo}.

\end{document}